\documentclass[a4paper,UKenglish,cleveref, autoref, thm-restate]{lipics-v2021}

\usepackage{todonotes}
\usepackage{microtype}
\usepackage{complexity}

\hideLIPIcs

\title{Higher Hardness Results for the Reconfiguration of Odd Matchings}

\author{Joseph Dorfer}{Graz University of Technology, Austria}{joseph.dorfer@tugraz.at}{https://orcid.org/0009-0004-9276-7870}{Austrian Science Foundation (FWF) 10.55776/DOC183}

\authorrunning{Joseph Dorfer}

\Copyright{Joseph Dorfer}



\nolinenumbers

\graphicspath{{./figures/}}

\acknowledgements{First, I would like to thank Lasse Wulf for introducing me to the relation between the polynomial hierarchy and diameter computation for flip graphs. Further, I would like to thank Oswin Aichholzer, Sofia Brenner, Hung Hoang, Daniel Perz, Christian Rieck, and Francesco Verciani for our previous joint work on odd matchings. Finally, I want to thank the reviewers for many helpful comments that helped improving the argumentation throughout the paper.}

\begin{document}
	
	\maketitle
	
	\begin{abstract}
		We study the reconfiguration of odd matchings of combinatorial graphs. Odd matchings are matchings that cover all but one vertex of a graph. A reconfiguration step, or flip, is an operation that matches the isolated vertex and, consequently, isolates another vertex. The flip graph of odd matchings is a graph that has all odd matchings of a graph as vertices and an edge between two vertices if their corresponding matchings can be transformed into one another via a single flip.
		
		We show that computing the diameter of the flip graph of odd matchings is $\Pi_2^p$-hard. This complements a recent result by Wulf [FOCS25] that it is~$\Pi_2^p$-hard to compute the diameter of the flip graph of perfect matchings where a flip swaps matching edges along a single cycle of unbounded size.
		
		Further, we show that computing the radius of the flip graph of odd matchings is $\Sigma_3^p$-hard. The respective decision problems for the diameter and the radius are also complete in the respective level of the polynomial hierarchy. This shows that computing the radius of the flip graph of odd matchings is provably harder than computing its diameter, unless the polynomial hierarchy collapses.
		
		Finally, we reduce \textsc{set cover} to the problem of finding shortest flip sequences. As a consequence, we show $\log$-\APX-hardness and that the problem cannot be approximated by a sublogarithmic factor. By doing so, we answer a question asked by Aichholzer, Brenner, Dorfer, Hoang, Perz, Rieck, and Verciani [GD25].

	\end{abstract}

	\newpage
	\section{Introduction}
	
	\emph{Reconfiguration} describes the process of changing one structure into another. It is often performed by small, reversible steps, so called \emph{flips}. Reconfiguration has many applications in the areas of optimization \cite{ItoDHPSUU11} or enumeration~\cite{avis1996reverse,mütze2024combinatorialgraycodesanupdated}. We refer to the following surveys for the discussion of additional applications of reconfiguration~\cite{Nishimura18,Heuvel13}. Very recently, reconfiguration has also provided substantial insight into the complexity of computing the diameter of polytopes~\cite{nobel2025complexity,sanita2018diameter}, spiking in the result that computing the combinatorial diameter of a polytope is $\Pi_2^p$-hard~\cite{wulf2025computingpolytopediameterharder}. Remarkably, all the recent results on the diameter of polytopes study the reconfiguration of perfect matchings.
	
	We study the related problem of the reconfiguration of \emph{odd matchings} of graphs. An odd matching of a graph $G$ is a matching consisting of edges of $G$ such that all vertices of $G$ are matched except for a single \emph{isolated vertex}. A \emph{flip} between two odd matchings is an operation that matches the isolated vertex $v$ of the first matching to another vertex $w$. Subsequently, the vertex that previously shared an edge with~$w$ becomes the new isolated vertex. The \emph{flip graph} of odd matchings of a graph $G$ is a graph that has as vertex set all odd matchings of $G$ and has edges between matchings whenever they can be transformed into one another via a single flip. A \emph{flip sequence} between an initial matching $M_{in}$ and a target matching $M_{tar}$ is a sequence of matchings $M_{in}=M_0$, $M_1$, $\ldots$, $M_{k-1}$,$M_{k}=M_{tar}$ such that consecutive matchings only differ by a single flip. In terms of the flip graph, a flip sequence is a path between two matchings. The index $k$ denotes the \emph{length} of a flip sequence. The \emph{flip distance} between two odd matchings $M_{in}$ and $M_{tar}$, denoted by $d(M_{in},M_{tar})$, is the minimum over all $k$ such that there exists a flip sequence of length $k$ between $M_{in}$ and $M_{tar}$. This can be interpreted as the length of a shortest path in the flip graph. The \emph{diameter} of the flip graph is defined as $\max_{M_1,M_2}~d(M_1,M_2)$ and describes the maximal flip distance between any pair of odd matchings in the flip graph. Similarly, the \emph{radius} is defined as $\min_{M_1}~\max_{M_2}~d(M_1,M_2)$ and describes the minimum maximal distance of a matching $M_2$ to a \emph{center} $M_1$. Clearly, the diameter of a flip graph is bounded from below by its radius and bounded from above by twice the radius. For an illustration of many of the concepts, we refer to Figure~\ref{fig:flipgraph}.
	
	\begin{figure}[h]{\columnwidth-10pt}
		\centering
		\includegraphics[scale=0.36]{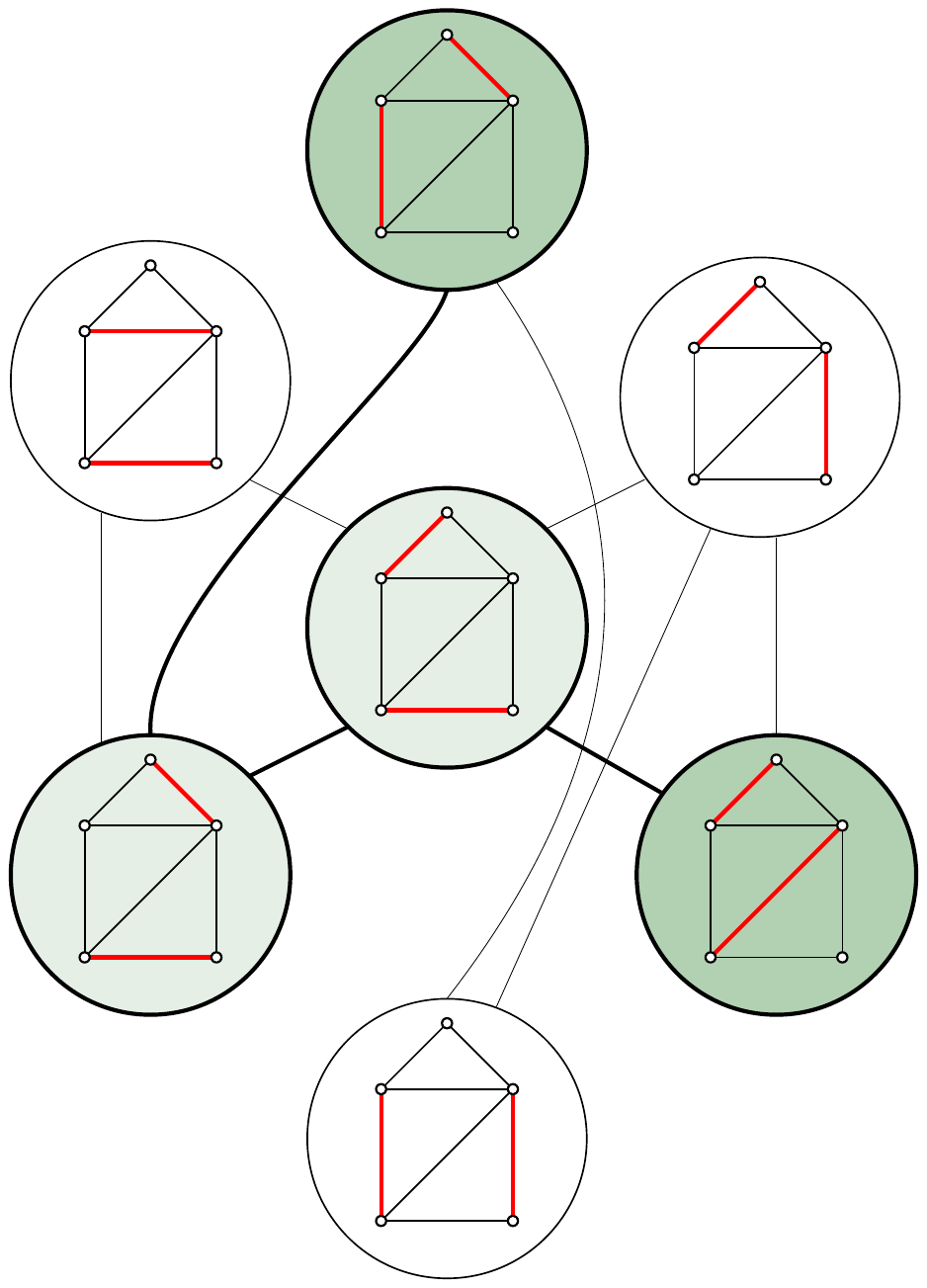}
		\caption{Illustration of a flip graph of odd matchings with a flip sequence of length $3$ highlighted.}%
		\label{fig:flipgraph}%
	\end{figure}%
	
	\subsection{Related Work}
	
	\subparagraph*{Odd Matchings} The reconfiguration of odd matchings has been introduced in~\cite{oddmatchings} for the setting of geometric odd matchings, that is, crossing-free odd matchings with straight line segments as edges between points in general position in the plane. The authors show that in this setting the flip graph is always connected with a diameter of $O(n^2)$ where $n$ is the number of points.  The study of combinatorial odd matchings has been studied in \cite{aichholzer2025flippingoddmatchingsgeometric}. The authors provide a complete, polynomial time checkable characterization when the flip graph of odd matchings on a graph $G$ is connected and show that any connected component of the flip graph has a diameter that is linear in the size of $G$. In the same paper, it is shown that the problem of deciding whether there exists a flip sequence of a certain length between two given odd matchings is \NP-complete for both the geometric and the combinatorial setting. For the combinatorial setting, the flip distance problem is already \NP-hard for solid rectangular grid graphs. The reconfiguration of odd matchings also appears in form of the sliding block game gourds~\cite{hamersma2020gourds,kakimura2024reconfiguration} where the underlying graph is a triangular grid graph and the edges have colored or labeled end points. In every flip, the added edge inherits the colors or labels of the removed edge. 
	
	\subparagraph*{Perfect Matchings} For the reconfiguration of perfect matchings, removing and adding a single edge will not yield  another perfect matching. Instead a flip in a matching $M$ works as follows: Pick a cycle that alternates between edges that lie in $M$ and edges that do not. The flip then removes all edges of $M$ along the cycle and adds the edges of the cycle that were not in $M$. We can either allow cycles of arbitrary length, or bound the length of the cycle. When bounding the length of the cycle to only allow cycles of length four such that a flip removes two edges from $M$ and adds two new edges it is shown to be \PSPACE-complete to decide whether a given matching can be flipped into another matching in the combinatorial setting \cite{bonamy2019perfectmatchingreconfigurationproblem}. In the geometric setting it is a long-standing open question whether any perfect matching on any point set can be transformed into any other matching on the same point set when only allowing flips along $4$-cycles. Up to now, there is not even a published proof that a perfect matching on any point set permits a valid flip of such a form. It further has been shown that deciding whether a flip sequence of a certain length exists is \NP-hard in the geometric setting \cite{binucci2025flippingmatchingshard}.
	
	If, however, we allow flips along alternating cycles of unbounded size, the flip graph is connected in the geometric setting~\cite{articlematchings_2}. In the combinatorial setting it has been shown that deciding whether the flip distance between two perfect matchings is at most $k$ is \NP-complete, even for $k=2$ \cite{aichholzer2021flip,ito2019shortest}. There also has been particular interest in the complexity of computing the diameter of the flip graph of perfect matchings in the combinatorial setting \cite{nobel2025complexity,sanita2018diameter,wulf2025computingpolytopediameterharder} spiking in the result that computing the diameter is $\Pi_2^p$-complete \cite{wulf2025computingpolytopediameterharder}. The research was motivated by its implications for the complexity of computing the diameter of polytopes \cite{chvatal1975certain}.
	
	\subparagraph{Computing Central Structures in Flip Graphs}
	
	This years CG:SHOP challenge deals with the search of central structures in flip graphs of triangulations under parallel flip operations~\cite{Cgshop}. The problem as it is phrased asks for a central structure in the flip graph for a small (compared to the size of the flip graph) finite set of triangulations that are part of the input. The problem is contained in \NP~and can be interpreted as an attempt to approximate the search of the center and the radius of the whole flip graph.
	
	\subsection{Our Contributions}
	
	In this paper, we provide three novel higher complexity results on the reconfiguration of combinatorial odd matchings.
	
	First, we complement the main result in \cite{wulf2025computingpolytopediameterharder} by showing that for a given graph $G$ and a parameter $k\in \mathbb{N}$ the problem of deciding whether the diameter of the flip graph of odd matchings of $G$ is at most $k$ is $\Pi_2^p$-complete. We do so by reducing directly from the $\Pi_2^p$-complete problem \textsc{$\forall\exists$-SAT}.
	
	\begin{restatable}{theorem}{PI}
		\label{thm:PI}
		Given a graph $G$ and a parameter $k\in \mathbb{N}$. Deciding whether the diameter of the flip graph of odd matchings of $G$ is at most $k$ is $\Pi_2^p$-complete.
	\end{restatable}
	
	As a second result, we study the related problem of calculating the radius of the flip graph. We show that deciding whether for a given graph $G$ the radius of the flip graph of odd matchings of $G$ is at most some value $k$ is $\Sigma_3^p$-complete. We do so by reducing directly from the $\Sigma_3^p$-complete problem \textsc{$\exists\forall\exists$-SAT}.
	
	\begin{restatable}{theorem}{SIG}
		\label{thm:SIGMA}
		Given a graph $G$ and a parameter $k\in \mathbb{N}$. Deciding whether the radius of the flip graph of odd matchings of $G$ is at most $k$ is $\Sigma_3^p$-complete.
	\end{restatable}
	
	By showing $\Sigma_3^p$-completeness of the problem, we provide a naturally occuring problem that is complete in this complexity class. As discussed in \cite{grüne2024completenesspolynomialhierarchynatural}, problems in this complexity class are not too well studied and the list of problems that are complete for that class are not too long.
	
	Even though the concepts of diameter and radius seem very similar, we conclude that it is provably harder to compute the radius than to compute the diameter, unless the polynomial hierarchy collapses.
	
	The authors of \cite{aichholzer2025flippingoddmatchingsgeometric} show that it is \NP-hard to compute shortest flip sequences between odd matchings and motivate the question about the existence of approximation algorithms.
	
	As a final result, we prove that computing the flip distance between two odd matchings is $\log$-\APX-hard\footnote{Only after submitting the camera-ready version of this paper we discovered paper \cite{bousquet2019shortest} which deals with an operation called \emph{token jumping}. The authors study a reconfiguration operation that transforms one matching into another. The matchings do not have to be perfect or almost perfect or even inclusion-wise maximal. A reconfiguration step removes an edge and adds another edge such that the resulting set of edges is again a matching. For matchings with a single isolated vertex the setting coincides with the setting in our paper. In~\cite{bousquet2019shortest} the authors show that the length of a shortest reconfiguration sequence cannot be efficiently approximated within a sublogarithmic factor, unless \P=\NP. The proof reduces from \textsc{set cover} and uses similar gadgets to ours. However, the matchings in the proof have one isolated vertex per set and the reduction in~\cite{bousquet2019shortest} is not PTAS-preserving and, thus, does not imply ($\log$-)\APX-hardness.} via a reduction from \textsc{set cover}~\cite{10.1145/2591796.2591884}.
	
	\begin{restatable}{theorem}{APXH}
		\label{thm:APX}
		Given a graph $G$ and two odd matchings $M_{in}$ and $M_{tar}$ of $G$. Computing the flip distance between $M_{in}$ and $M_{tar}$ is $\log$-\APX-hard. In particular, it is NP-hard to approximate the flip distance by a factor better than $\Theta(log(n))$.
	\end{restatable}
	
	Full technical details and proofs for statements marked with $(\star)$ will appear in a full version of the paper.
	
	\section{Preliminaries}
	
	\subsection{Union of Odd Matchings}
	
	Let $M_1$ and $M_2$ be two odd matchings on the same graph $G$, then their union $M\cup M'$ admits the following connected components:
	\begin{itemize}
		\item one \emph{alternating path} of even length (possibly zero) that connects the isolated vertices of $M_1$ and $M_2$ and alternates between edges of $M_2$ and $M_1$,
		\item \emph{cycles} of even length alternating between $M_1$ and $M_2$, and
		\item edges that lie in $M_1\cap M_2$, called \emph{happy edges}.
	\end{itemize}
	
	The above partition will be helpful when arguing lower bounds on the length of flip sequences. We will say that we \emph{charge flips towards a component} if a flip sequence needs to perform that number of flips on this connected component. Some easy observations are: (1) If the alternating path contains $k$ edges of $M_1$ and $M_2$ each, then we charge at least $k$ flips towards the path; (2) If an alternating cycle contains $k$ edges of $M_1$ and $M_2$ each, we charge at least $k+1$ flips towards the cycle, $k$ to flip all the edges of $M_2$ in and one additional flip to place the isolated vertex on the cycle, we will call these steps \emph{switching a cycle} and (3) we either charge no flips at all towards a happy edge or at least two since if we remove the happy edge, we need to add it back in.
	
	\begin{figure}[ht]
		\centering
		\includegraphics[scale=0.6]{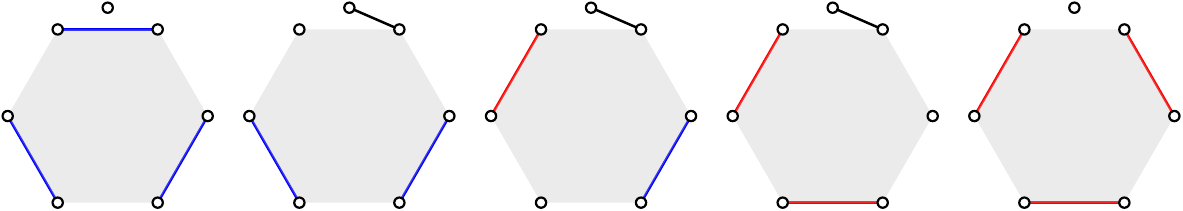}
		\caption{Switching a cycle with three matching edges in four flips.}
		\label{fig:switch}
	\end{figure}
	
	\subsection{The Polynomial Hierarchy}
	
	The polynomial hierarchy was introduced by Stockmeyer \cite{STOCKMEYER19761} and provides a way to compare the complexity of problems beyond \NP-hardness. Complexity classes are defined recursively. The lowest level is $\Sigma_0^p = \Pi_0^p = P$. Then, $\Sigma_k^p$ for $k\geq 1$ is defined as the class of all problems that can be decided in non-deterministic polynomial time with the help of an oracle for the class $\Sigma_{k-1}^p$. Further, $\Pi_k^p = co\text{-}\Sigma_k^p$. In particular $\Sigma_1^p = \NP$ and $\Pi_1^p=co\text{-}\NP$.
	
	We refer to \cite{wrathall1976complete} for a definition of the polynomial hierarchy that is easier to work with in our setting. A \emph{language} is a set $L\subseteq \{0,1\}^\ast$. A language $L$ is contained in $\Pi_2^p$ if there exists some polynomial-time computable function $V$ such that for all $w \in \{0,1\}^\ast$ for suitable $m_1,m_2 \leq poly(\lvert w \rvert)$:
	\begin{equation*}
		w \in L \Leftrightarrow \forall y_1 \in \{0,1\}^{m_1} \exists y_2 \in \{0,1\}^{m_2}: V(w,y_1,y_2) = 1
	\end{equation*}
	
	Similarly, a language $L$ is contained in $\Sigma_3^p$ if there exists some polynomial-time function~$V$ such that for all $w \in \{0,1\}^\ast$ for some suitable $m_1,m_2,m_3 \leq poly(\lvert w \rvert)$:
	
	\begin{equation*}
		w \in L \Leftrightarrow \exists y_3 \in \{0,1\}^\ast \forall y_2 \in \{0,1\}^{m_2} \exists y_3 \in \{0,1\}^{m_3}: V(w,y_1,y_2,y_3) = 1
	\end{equation*}
	
	In \cite{wrathall1976complete} problems are provided that are complete for the respective stages of the polynomial hierarchy. The problem \textsc{$\forall \exists$-SAT} given by all Boolean formulas $\phi$ on variables $x_1,...,x_{m_1}$ and $y_1,...,y_{m_2}$ such that for all assignments of $x_1,...,x_{m_1}$ there exists an assignment of $y_1,...,y_{m_2}$ such that $\phi(x_1,...,x_{m_1},y_1,...,y_{m_2})=1$ is $\Pi_2^p$-complete. Further, the problem \textsc{$\exists\forall\exists$-SAT} given by all Boolean formulas $\psi$ on variables $x_1,...,x_{m_1},y_1,...,y_{m_2},z_1,...,z_{m_3}$ such that there exists an assignment of $x_1,...,x_{m_1}$ such that for all assignments of $y_1,...,y_{m_2}$ there exists an assignment of $z_1,...,z_{m_3}$ such that $\psi(x_1,...,x_{m_1},y_1,...,y_{m_2},z_1,...,z_{m_3})=1$ is $\Sigma_3^p$-complete. We assume without loss of generality that all Boolean formulas are given in conjunctive normal form (CNF).
	
	\subsection{APX-hardness}
	
	For an extensive introduction to the concepts see \cite{crescenzi}. Let \APX~be the set of all problems in \NP~that allow for a constant factor approximation. Further, let $\log$-\APX~describe the class of problems that allow an approximation by a factor that is logarithmic in the input size. A problem is \emph{\APX-hard} if there is a \emph{PTAS-reduction} from every problem in \APX~to said problem. A PTAS-reduction from problem $A$ to $B$ is a set of three functions $f$, $g$, $\alpha$, that are polynomial-time computable for a fixed $\varepsilon$ such that:
	\begin{itemize}
		\item the function $f$ maps an instance of $A$ to an instance of $B$.
		\item the function $g$ takes an instance $x$ of $A$ and an approximate solution of $f(x)$ in $B$ and computes an approximation of $x$
		\item  the function $\alpha$ maps error parameters of problems in $A$ to corresponding parameters of problems in~$B$
		\item if the solution $y$ to $f(x)$ is an $1+\alpha(\varepsilon)$-approximation to the optimal solution, then $g(x,y,\varepsilon)$ is a $1+\varepsilon$ solution to $x$.
	\end{itemize}
	In particular, if there exists no polynomial-time approximation scheme for $A$ and $B$ PTAS-reduces to $A$, then there is also no polynomial-time approxiation scheme for $B$ and if $A$ cannot be approximated within some factor $1+\varepsilon$ then $B$ cannot be approximated within~$1+\alpha(\varepsilon)$.
	
	A PTAS-reduction is called an \emph{AP-reduction} if, additionally, there exists a constant $\beta$ such that whenever $y$ is an $r$ approximation for $f(x)$ then $g(x,y,\varepsilon)$ is a $1+\beta(r-1)$ approximation for~$x$. If additionally $f$ and $g$ do not depend on the choice of $\varepsilon$ then the AP-reduction preserves $\log$-APX membership. A problem is $\log$-APX hard if there exists an AP-reduction for which $f$ and $g$ are independent of $\varepsilon$ from any problem in $\log$-APX to said problem.
	
	\subsection{Set Cover}
	
	Consider the integers from $1$ to $n$ and a collection of sets $S = \{s_1,\ldots,s_t\}$ such that $s_i\subseteq \{1,\ldots n\}$ and $\bigcup_{i=1}^t s_i = \{1,\ldots,n\}$. The \textsc{Set Cover} problem asks for a given integer~$k$ whether there exists a subset $S' \subseteq S$ such that $\lvert S' \rvert = k$ and $\bigcup_{s_i \in S'} s_i = \{1,\ldots,n\}$. \textsc{Set Cover} is known to be $\log$-\APX-hard. Further, the size of a smallest set $S^\ast$ that covers all integers in~$\{1,\ldots,n\}$ can in general not be approximated by a sublogarithmic factor \cite{10.1145/2591796.2591884}, unless~\P=\NP.
	
	\section{Computing the Diameter is $\Pi_2^p$-complete}\label{sec:diameter}
	
	We reduce directly from \textsc{$\forall \exists$-SAT} to computing the diameter of the flip graph. We will introduce gadgets for clauses and for each type of variable.
	
	 As a high level idea, the goal of a flip sequence will be to switch the cycles in all clause gadgets by switching the cycles in all the variable gadgets in a way that places isolated vertices next to clause gadgets. The alternating cycle of an $\exists$-gadget can then be switched in two ways depending on which of the two edges from a central vertex $v$ to the gadget is used to enter the gadget. Each direction for the switch will correspond to an assignment of the variable based on what clauses the isolated vertex is placed next to. In a $\forall$-gadget, there is only one way to enter the gadget from $v$, so the way to switch the cycle is fixed. If the cycle is not crossing in the drawing, the isolated vertex will be placed next to gadgets of clauses that contain the positive literal, and, if the cycle is crossing, next to clauses containing the negative literal. 
	
	We now introduce all gadgets of our reduction in more detail. The reduction will be built around one, aforementioned, central vertex $v$.
	
	\begin{figure}[htb]
		\begin{subfigure}[t]{\dimexpr0.3\columnwidth-10pt\relax}
			\centering
			\includegraphics[scale=0.3]{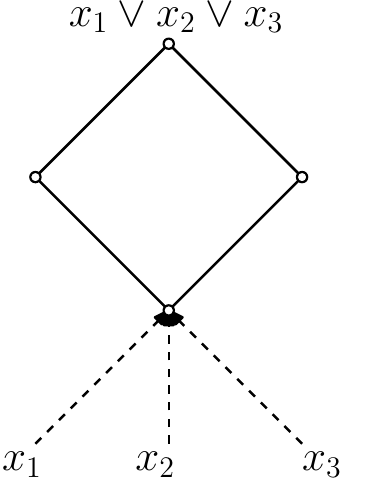}
			\subcaption{The clause Gadget.}%
			\label{fig:clause}%
		\end{subfigure}%
		\hfill%
		\begin{subfigure}[t]{0.3\columnwidth-10pt}
			\centering
			\includegraphics[page=1, scale=0.3]{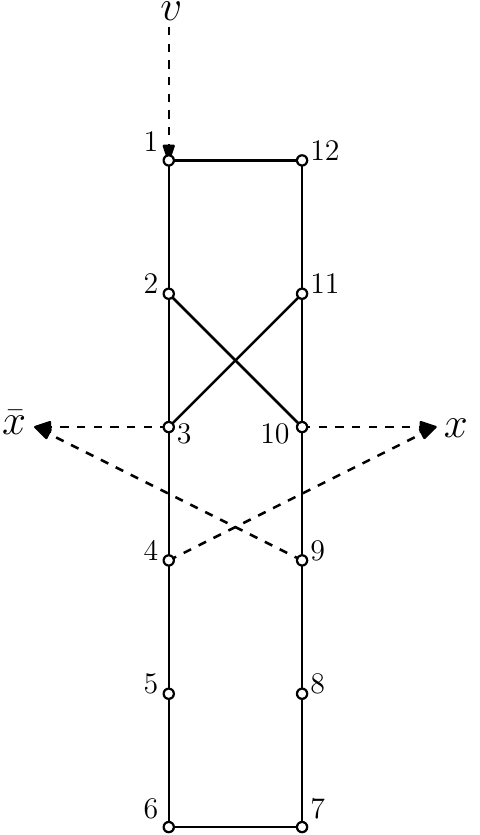}
			\subcaption{The $\forall$-gadget}%
			\label{fig:forall}%
		\end{subfigure}%
		\hfill%
		\begin{subfigure}[t]{0.3\columnwidth-10pt}
			\centering
			\includegraphics[scale=0.3]{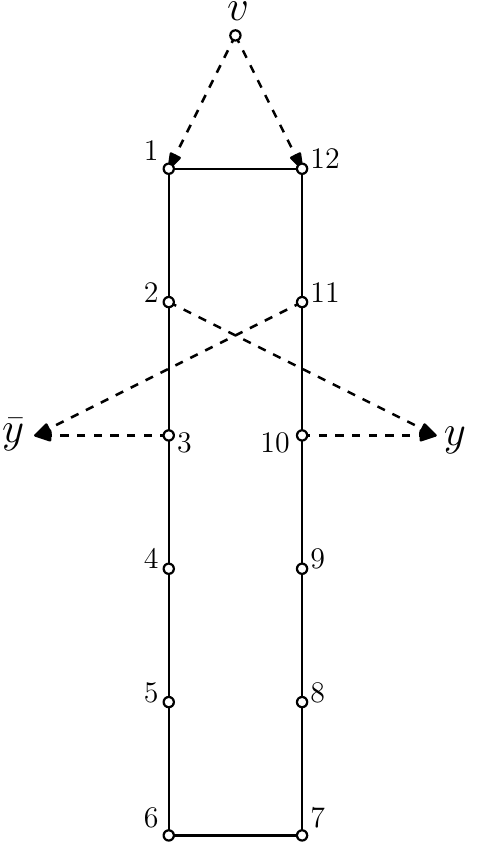}
			\subcaption{The $\exists$-gadget}%
			\label{fig:exists}%
		\end{subfigure}%
		\caption{Gadgets of the reduction}
	\end{figure}
	
	\begin{description}
		\item[Clause gadget (Figure \ref{fig:clause}):] The clause gadget is a 4-cycle that has one vertex which is incident to the later introduced variable gadgets. The idea is to force a perfect matching of the cycle in the initial matching and the other perfect matching in the target matching such that all alternating cycles in variable gadgets need to be switched at some point.
		
		\item[$\forall$-gadget (Figure \ref{fig:forall}):] For a given variable $x$, the variable gadget consists of a cycle of length twelve, with vertices labeled $1$ to $12$ along the cycle, as well as two additional diagonals from vertex $2$ to vertex $10$ and vertex $3$ to vertex $11$ in the cycle that form the crossing as seen in Figure~\ref{fig:forall}. The vertex with label $1$ has an edge that is connected to a central vertex $v$. Two vertices, labeled $3$ and $9$, have edges to all clause gadgets that correspond to clauses that contain $\bar{x}$. Two vertex, labeled $4$ and $10$, are incident to edges to all clause gadgets that correspond to clauses that contain $x$. If $M_{in}\cup M_{tar}$ contains an alternating cycle on a $\forall$-gadget, then this cycle can either contain the introduced crossing or not. To see a crossed and uncrossed cycle, see Figure~\ref{fig:cforall}. The two options will encode the truth value of the corresponding variable.
		
		\begin{figure}[ht]{\columnwidth-10pt}
			\centering
			\includegraphics[page=2, scale=0.3]{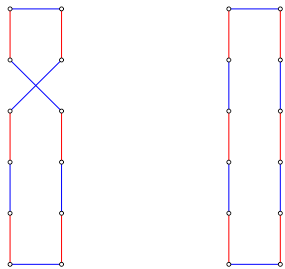}
			\caption{Left: a crossed alternating cycle in a $\forall$-gadget, Right: an uncrossed alternating cycle.}%
			\label{fig:cforall}%
		\end{figure}%
		
		\item[$\exists$-gadget (Figure \ref{fig:exists}):] The $\exists$-gadget corresponding to a variable $y$ consists of a cycle with twelve edges and vertices labeled $1$ to $12$ along the cycle. Two adjacent vertices, labeled $1$ and $12$ are incident to an edge to the vertex $v$. Two vertices, labeled $2$ and $10$ are incident to edges to all gadgets of clauses that contain $y$. Similarly, two vertices, labeled $3$ and $11$ are connected to gadgets of clauses that contain $\bar{y}$. If the union $M_{in}\cup M_{tar}$ contains an alternating cycle on an $\exists$-gadget, the cycle can be switched in two ways depending on which of the two edges incident to $v$ is used to switch the cycle. These two choices will encode the two truth values of the corresponding variable.
		
		\item[Forcing the position of the isolated vertex:] If we take $s$ to be the number of vertices in all clause gadgets, $\forall$-gadgets and $\exists$-gadgets combined, we obtain by \cite[Theorem 10]{aichholzer2025flippingoddmatchingsgeometric} that any connected component of the flip graph of odd matchings of the constructed graph has diameter at most $c\cdot (s+1)$ for a constant $c$. We attach to $v$ a path $P$ of length $\ell=2c\cdot(s+1)$ (See Figure \ref{fig:instance}). If the isolated vertex of at least one of the matchings is placed at the vertex $w$ of $P$ that is  farthest away from $v$, then flipping edges along the path already takes $\frac{\ell}{2}$ flips, which is at least as much as it takes to reconfigure two odd matchings if both have their isolated vertex placed on some clause gadget, $\forall$-gadget, or $\exists$-gadget.
		
		\item[The reduction:] Now, let $\phi$ be a Boolean formula on variables $x_1$,...,$x_{m_1}$ and $y_1,...,y_{m_2}$. We construct a graph $G$ as follows: Start from a single vertex $v$. For every $x_i$ introduce a $\forall$-gadget and connect it to $v$ as described above. For every~$y_i$ introduce an $\exists$-gadget and connect it to $v$ as described. Further, introduce a clause gadget for every clause $C$ in~$\phi$ and connect each clause gadget to all the variable gadgets that correspond to variables in the clause. The connection happens in one of two ways described above depending on whether $x \in C$ or $\bar{x} \in C$ (resp. $y$ or $z$). Then add the path of length $\ell$ at $v$ to obtain a final graph $G'$.
	\end{description}
	
	\begin{figure}[ht]
		\centering
		\includegraphics[scale=0.5]{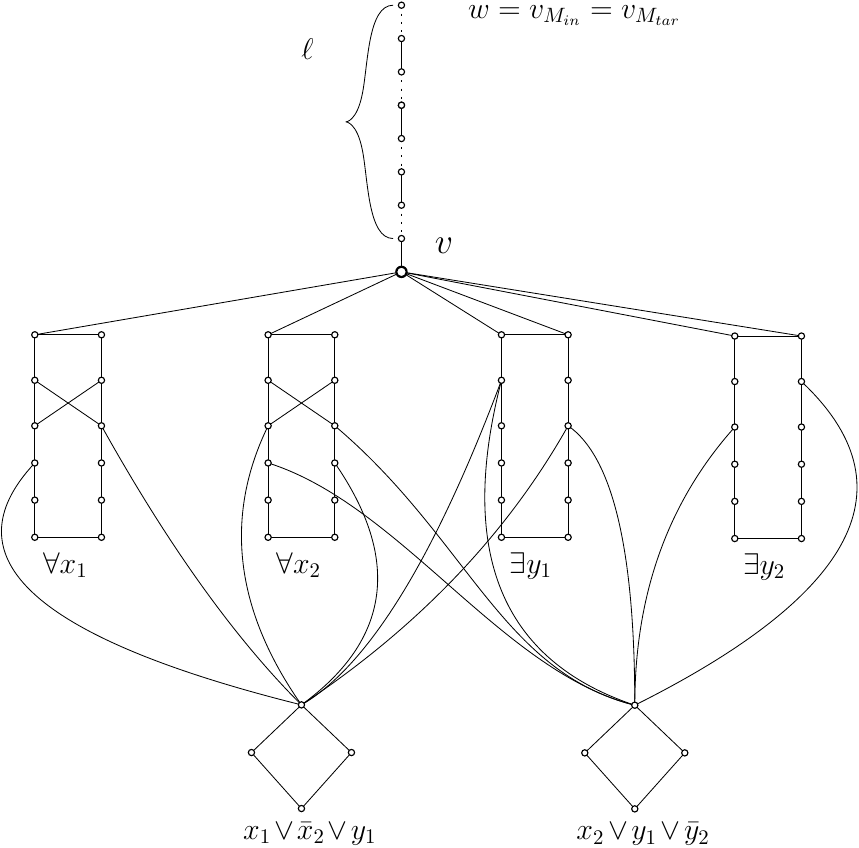}
		\caption{The instance graph for the quantified formula $\forall x_1,x_2~\exists y_1,y_2: (x_1 \lor \bar{x}_2 \lor y_1)\land(x_2\lor y_1\lor\bar{y}_2)$}
		\label{fig:instance}
	\end{figure}
	
	\newpage
	
	First, we need to make sure that the diameter is well defined in this setting.
	
	\begin{restatable}[$\star$]{proposition}{connected}
		\label{prop:connected}
		The flip graph of odd matchings of $G'$ is connected.
	\end{restatable}
	
	We show Proposition~\ref{prop:connected} by showing that $G'$ fulfills the characterization of graphs for which the flip graph is connected from~\cite[Theorem 3]{aichholzer2025flippingoddmatchingsgeometric}.
	
	From now on, we will call a pair of matchings $M_{in}$ and $M_{tar}$ such that $d(M_{in},M_{tar})$ equals the diameter of the flip graph a \emph{maximizing pair}. We will make structural observations on how a maximizing pair looks like. Also, let $v_{M_{in}}$ and~$v_{M_{tar}}$ denote the isolated vertices of $M_{in}$ and $M_{tar}$.
	
	\begin{restatable}[$\star$]{lemma}{structure}
		\label{lem:structure}
		Let $M_{in}$ and $M_{tar}$ be a maximizing pair. Then:
		\begin{enumerate}[(1)]
			\item $v_{M_{in}}=v_{M_{tar}}$ is the vertex $w$ of $P$ that is  farthest away from $v$.
			\item $M_{in}\cup M_{tar}$ contains an alternating cycle on every clause gadget.
			\item $M_{in}\cup M_{tar}$ contains an alternating cycle on every $\exists$-gadget.
			\item $M_{in}\cup M_{tar}$ contains an alternating cycle on every $\forall$-gadget.
		\end{enumerate}
	\end{restatable}
	
	\PI*
	
	\begin{proof}
		Containment in $\Pi_2^p$ follows from the definition of $\Pi_2^p$ and the fact that the underlying flip distance problem that has to be solved for every pair of matchings is contained in \NP.
		
		Let $\phi$ on variables $x_1,...,x_{m_1}$ and $y_1,...,y_{m_2}$ and clauses $C_1,...,C_K$ be a $\forall\exists$-SAT instance and $G'$ be the graph as constructed above.
		
		\begin{restatable}{claim}{hin}
			\label{claim:hin}
			If $\phi$ is a \textsc{YES}-instance of \textsc{$\forall\exists$-SAT}, then the diameter of the flip graph of odd matchings of $G'$ is at most $\ell+7(m_1+m_2)+3K$.
		\end{restatable}
		
		\begin{claimproof}
			Assume $\phi$ is a \textsc{YES}-instance. In order to upper bound the diameter of the flip graph of odd matchings of $G'$ it is sufficient to consider pairs of matchings that fulfill the conditions of Lemma~\ref{lem:structure}. Let $M_{in}$ and~$M_{tar}$ be a pair with these properties. We consider $M_{in}\cup M_{tar}$ on every $\forall$-gadget for some variable~$x_i$. If they form an uncrossed cycle, then we consider the positive assignment of $x_i$, otherwise the negative assignment. Since $\phi$ is a \textsc{YES}-instance of \textsc{$\forall\exists$-SAT}, for the given assignment of the $x_i$ there exists a satisfying assignment of the $y_i$ such that $\phi(x_1,...,x_{m_1},y_1,...,y_{m_2})=1$. We can construct a flip sequence from $M_{in}$ to $M_{tar}$ based on the assignment of $y_1,...,y_{m_2}$.
			
			The flip sequence starts by moving the isolated vertex from $w$ to $v$. Then the flip sequence switches all cycles that belong to $\forall$-gadgets and if the isolated vertex is placed next to an unswitched clause gadget, the flip sequence switches this cycle as well. Afterwards, the flip sequence switches all $\exists$-gadgets. If $y_i$ is assigned a positive (resp. negative) value, then the flip sequence traverses the gadget such that the isolated vertex is placed next to clause gadgets containing $y_i$ (resp. $\bar{y_i}$). Since $\phi(x_1,...,x_{m_1},y_1,...,y_{m_2})=1$, every clause gadget will eventually be switched. The flip sequence ends by moving the isolated vertex back from $v$ to~$w$. There are $\ell$ flips for moving the isolated vertex between $w$ and $v$, seven flips per variable gadget and three flips per clause gadget, which gives us the aspired length of the flip sequence.  
		\end{claimproof}
		
		\begin{restatable}{claim}{ruck}
			\label{claim:ruck}
			$\phi$ is a \textsc{YES}-instance of \textsc{$\forall\exists$-SAT} whenever the diameter of the flip graph of odd matchings of $G'$ is at most $\ell+7(m_1+m_2)+3K$.
		\end{restatable}
		
		\begin{claimproof}
			Assume the diameter of the flip graph of odd matchings is at most $\ell+7(m_1+m_2)+3K$. For a given assignment of $x_1,...,x_{m_1}$, we construct a pair of matchings $M_{in}$ and $M_{tar}$. If~$x_i$ appears in its positive form, $M_{in}\cup M_{tar}$ has an uncrossed cycle in the $\forall$-gadget that corresponds to $x_i$, otherwise $M_{in}\cup M_{tar}$ has a crossed cycle. Additionally, $M_{in}\cup M_{tar}$ has an alternating cycle on all $\exists$-gadgets and clause gadgets, and $v_{M_{in}}=v_{M_{tar}}=w$.
			
			Since there exists a flip sequence from $M_{in}$ to $M_{tar}$ of length at most $\ell+7(m_1+m_2)+3K$, exactly seven flips are performed in every variable gagdet and exactly three flips are performed in every clause gadget. Less flips would not suffice to switch the cycle in the respective gadget. With more flips we would exceed the length of the flip sequence. This means there is an isolated vertex placed next to every clause gadget at some point in a flip sequence that switches every variable gadget once. We construct an assignment of $y_1,...,y_{m_2}$ from where the isolated vertex was placed during the traversal of the corresponding variable gadgets. This assignment satisfies $\phi(x_1,...,x_{m_1},y_1,...,y_{m_2})=1$. Repeating this procedure for all assignments of $x_1,...,x_{m_1}$ shows that $\phi$ is a \textsc{YES}-instance of \textsc{$\forall\exists$-SAT}.  
		\end{claimproof}
		
		The theorem follows from a combination of the two claims.
	\end{proof}
	
	
	\section{Computing the Radius is $\Sigma_3^p$-complete}\label{sec:radius}
	
	We reduce from \textsc{$\exists\forall\exists$-SAT} to computing the radius of the flip graph. We now introduce all gadgets for our reduction that will again be built around one central vertex $v$.
	
	\begin{figure}[ht]
		\centering
		\includegraphics[scale=0.35]{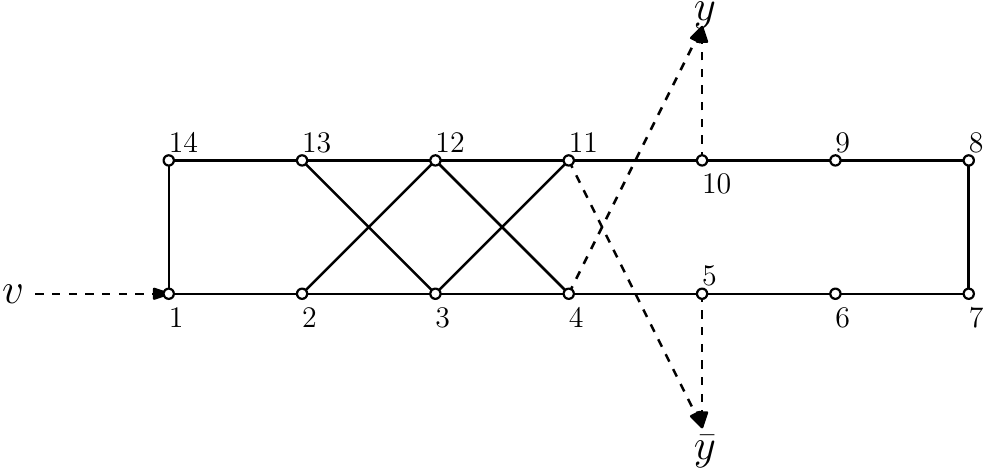}
		\caption{$\forall$-gadget for the reduction from $\exists\forall\exists$-\textsc{SAT} to computing the radius}
		\label{fig:2forall}
	\end{figure}
	
	\begin{description}
		\item[Clause gadget (Figure \ref{fig:clause}):] The clause gadget is the same as in Section~\ref{sec:diameter}.
		
		\item[First $\exists$-gadget (Figure \ref{fig:forall}):] The first $\exists$-gadget is the same as the $\forall$-gadget of Section~\ref{sec:diameter}.
		
		\item[$\forall$-gadget (Figure \ref{fig:2forall}):] The $\forall$-gadget corresponding to a variable $y$ consists of a cycle with $14$ edges and vertices labeled $1$ to $14$ along the cycle, as well as four diagonals, from $2$ to $12$, $3$ to $11$, $4$ to $12$, and $3$ to $13$, of the cycle that form the two crossings in Figure~\ref{fig:2forall}. Two vertices, $4$ and $10$, have edges to clauses that contain the positive literal $y$ and two vertices, $5$ and $11$, with edges to clauses containing the negative literal $\bar{y}$ and vertex $1$ has an edge from the vertex $v$.
		
		\item[Second $\exists$-gadget (Figure \ref{fig:exists}):] The second exists gadget coincides with the $\exists$-gadget of Section~\ref{sec:diameter}. Notation wise, the variables corresponding to those gadgets will now be called~$z$ instead of $y$.

		\begin{figure}[ht]
			\centering
			\includegraphics[scale=0.5]{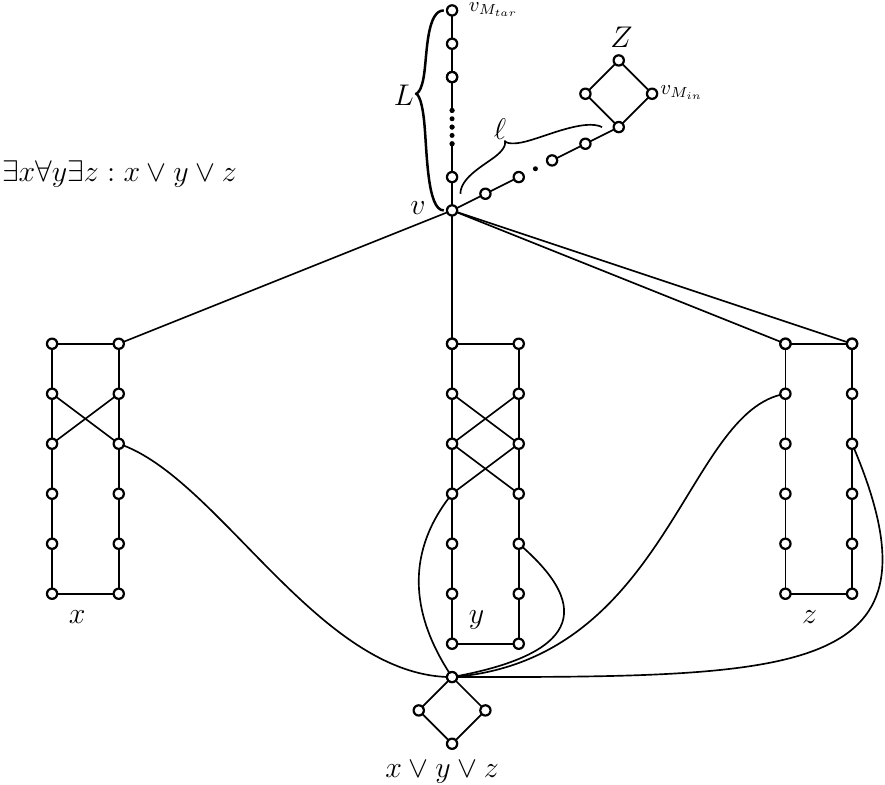}
			\caption{Instance graph for the quantified formula $\exists x~\forall y~\exists z: x \lor y \lor z$.}
			\label{fig:radius}
		\end{figure}
		
		\item[Forcing position of the isolated vertices:] We conclude the construction by attaching two paths to $v$. Let $s$ be the number of vertices in clause and variable gadgets. We set $\ell = 2c\cdot(s+1)$ where $c$ is chosen according to \cite[Theorem~10]{aichholzer2025flippingoddmatchingsgeometric} such that $\ell$ is larger than twice the diameter of the flip graph of odd matchings on $v$ and the clause and variable gadgets. We add a path $p$ of length $\ell$ and attach one end to $v$ and to the other end we attach a four-cycle $Z$. Further, let $L = 2c\cdot(s+\ell+1)$ and add a path $P$ of length $L$ and attach one end of $P$ to $v$. No matter what a matching $M_{in}$ in the center of the flip graph looks like, a matching $M_{tar}$ can maximize the distance to $M_{in}$ by having an isolated vertex at $v_{M_{tar}}$, the end vertex of $P$ that is not attached to $v$ and by differing from $M_{in}$ on $Z$. By doing so, any flip sequence is forced to traverse both $p$ and $P$. Vice versa, by having the isolated vertex of $M_{in}$ placed on $Z$ it can be guaranteed that a flip sequence traverses $p$ only once instead of twice.
		
		\item[The reduction:] Now, let $\psi$ be a Boolean formula on variables $x_1$,...,$x_{m_1}$, $y_1,...,y_{m_2}$ and $z_1,...,z_{m_3}$. We construct a graph $G$ as follows: Start from a single vertex $v$. For every $x_i$ introduce a first $\exists$-gadget and connect it to $v$ as described above. Repeat for all $y_i$ and $\forall$-gadgets and all $z_i$ and the second $\exists$-gadget. Further, for every clause $C$ in $\psi$ introduce a clause gadget and connect the gadget to all the variable gadgets that correspond to variables in the clause. The connection happens in one of two described ways depending whether $x \in C$ or $\bar{x} \in C$ (resp.~$y$ or $z$). Then add the paths $p$ and $P$ as described to obtain a final graph $G''$. Again, we need to make sure that the radius is well defined in this setting.
	\end{description}
	
	\begin{restatable}[$\star$]{proposition}{connecteds}
		\label{prop:sconnected}
		The flip graph of odd matchings on $G''$ is connected.
	\end{restatable}
	
	Again, we verify this by checking the characterization from~\cite[Theorem 3]{aichholzer2025flippingoddmatchingsgeometric}.
	
	We will again give a characterization of pairs of matchings, which can determine the radius. We will argue that, without loss of generality, we do not have to consider all other pairs of matchings but only a set of candidates. As a high level argument, we make sure that in a candidate pair we cannot make a change to $M_{tar}$ locally in a single gadget such that the flip distance to $M_{in}$ increases and we cannot immediately make a change to $M_{in}$ in a single gadget such that the flip distance decreases even if $M_{tar}$ is changed locally to respond to the changes to $M_{in}$.
	
	\begin{restatable}[$\star$]{lemma}{structural}
		\label{lem:lstructural}
		Without loss of generality, to determine the radius of the flip graph, it is sufficient to consider pairs of matchings $M_{in}$ and $M_{tar}$ that have the following properties:
		\begin{enumerate}[(1)]
			\item $v_{M_{tar}}$ is the vertex on $P$ that is farthest away from $v$.
			\item $v_{M_{in}}$ is placed on $Z$.
			\item $M_{in}\cup M_{tar}$ contains an alternating cycle on every clause gadget.
			\item $M_{in}\cup M_{tar}$ contains an alternating cycle on every variable gadget.
		\end{enumerate}
	\end{restatable}
	
	\vspace{11pt}
	
	With Lemma \ref{lem:lstructural} in mind, we describe the high level strategies for $M_{in}$, $M_{tar}$ and a shortest flip sequence:
	
	Since all clause gadgets contain alternating cycles, a shortest flip sequence will have to switch them all by entering them through a variable gadget. Adjacency relations of gadgets correspond to containment relations of variables in clauses.
	
	For the first $\exists$-gadget, while switching an uncrossed cycle, the isolated vertex will be next to clause gadgets containing the positive literal, otherwise, if the cycle is crossed, the clauses with negative literals can be switched. Also note, that $M_{in}$ can determine, whether the cycle is crossed or uncrossed. $M_{in}$ may choose not to control the existence of the crossing if the value of a particular variable does not matter.
	
	For the $\forall$-gadget, while switching the cycle, the isolated vertex is next to clauses with positive literals if the cycle is crossing twice or not at all, however if it crosses once, the isolated vertex is placed next to clauses with negative literals. While $M_{in}$ can control to add one crossing to the cycle or not, the existence of a second crossing will then always be controlled by $M_{tar}$.
	
	For the second $\exists$-gadget, there is only one option for an alternating cycle up to swapping $M_{in}$ and $M_{tar}$. The direction of traversal is then chosen by the flip sequence depending on which of the two edges incident to $v$ is chosen to enter the gadget. The direction of traversal then determines whether the isolated vertex will be placed next to clauses with positive or negative literals.
	
	\SIG*
	
	\begin{proof}
		Containment in $\Sigma_3^p$ follows from the definition of $\Sigma_3^p$ and the fact that the underlying flip distance problem that has to be solved for every matching~$M_{tar}$ is contained in \NP.
		
		Let $\psi$ on variables $x_1,...,x_{m_1}$, $y_1,...,y_{m_2}$ and $z_1,...,z_{m_3}$ and clauses $C_1,...,C_K$ be a \textsc{$\exists\forall\exists$-SAT} instance and $G''$ be the graph as constructed above.
		\begin{restatable}{claim}{hiin}
			\label{claim:hiin}
			If $\psi$ is a \textsc{YES}-instance of \textsc{$\exists\forall\exists$-SAT}, then the radius of the flip graph of odd matchings of $G''$ is at most $\frac{L+\ell}{2}+7(m_1+m_3)+8m_2+3K+2$.
		\end{restatable}
		
		\begin{claimproof}
			We construct $M_{in}$ based on the truth assignment of $x_1,...,x_{m_1}$ that is part of a solution of $\psi$. If $x_i$ has a positive truth value, then we match its corresponding variable gadget such that any completion of the matching to an alternating cycle is uncrossed. If $x_i$ has a negative truth value, match the variable gadget such that the cycle has a crossing. Place the isolated vertex on $Z$ and match the remaining vertices to form an arbitrary perfect matching. The existence of such perfect matchings is shown as a part of the proof of Proposition~\ref{prop:sconnected}.
			
			Now for any $M_{tar}$ that maximizes the distance to $M_{in}$, we look at the $\forall$-gadgets. By Lemma~\ref{lem:lstructural} we know that $M_{in}\cup M_{tar}$ forms an alternating cycle on the $\forall$-gadget. If the cycle in the gadget to some $y_i$ is crossed zero times or twice, we translate this to a positive assignment of $y_i$, if it is crossed once, translate to a negative assignment of $y_i$.
			
			Since $\psi$ is a \textsc{YES}-instance of $\exists\forall\exists$-SAT there exists an assignment of $z_1,...,z_{m_3}$ such that combined with the initial assignment of $x_1,...,x_{m_1}$ and the translated assignment of $y_1,...,y_{m_2}$ it holds that $\psi(x_1,...,x_{m_1},y_1,...,y_{m_2},z_1,...,z_{m_3})=1$. We build our flip sequence between $M_{in}$ and $M_{tar}$ around the assignment of the~$z_i$'s.
			
			A flip sequence from $M_{in}$ to $M_{tar}$ looks as follows: Perform flips to make the matchings coincide on $Z$ in up to two flips (since the isolated vertex is already placed on $Z$). Move the isolated vertex along $p$ to $v$ in $\frac{\ell}{2}$ flips. Switch all of the first $\exists$-gadgets using seven flips per gadget and if the isolated vertex is placed next to an unswitched clause gadget, switch it using three flips. Switch all $\forall$-gadgets using eight flips per gadget and if the isolated vertex is placed next to an unswitched clause gadget, switch it using three flips. For every $\exists$-gadget that corresponds to some $z_i$ switch the $\exists$-gadget such that if $z_i$ has a positive truth value, the isolated vertex is placed next to clauses that contain the positive literal $z_i$. However, if $z_i$ has a negative truth value, the isolated vertex should be placed next to clauses that contain the negative literal $\bar{z}_i$. Every switch of a $\exists$-gadget takes again seven flips. In the end, move the isolated vertex along $P$ using $\frac{L}{2}$ flips.
			
			The total length of the flip sequence adds up to $\frac{L+\ell}{2}+7(m_1+m_3)+8m_2+3K+2$.  
		\end{claimproof}		
		
		\begin{restatable}{claim}{ruuck}
			\label{claim:ruuck}
			$\psi$ is a \textsc{YES}-instance of \textsc{$\exists\forall\exists$-SAT} whenever the radius of the flip graph of odd matchings of $G''$ is at most $\frac{L+\ell}{2}+7(m_1+m_3)+8m_2+3K+2$.
		\end{restatable}
		
		\begin{claimproof}
			Let $M_{in}$ be the center of the flip graph of matchings of $G''$. In particular, the flip distance from $M_{in}$ to any other odd matching of $G''$ is at most the radius.
			
			Consider the first set of variable gadgets corresponding to $x_1,...,x_{m_1}$ and see, in which way they are matched. If the gadget corresponding to $x_i$ is matched, such that any completion to an alternating cycle is uncrossed, assign a positive truth value to~$x_i$, if any completion to an alternating cycle is crossed, assign a negative value. It can happen that the matching can be clompeted to both a crossed and uncrossed alternating cycle, in that case assign an arbitrary truth value.
			
			Now, let there be an arbitrary truth assignment of $y_1,...,y_{m_2}$. We construct~$M_{tar}$ based on this assignment. In every $\forall$-gadget, we complete $M_{in}$ to an alternating cycle that has zero or two crossings if the corresponding $y_i$ has a positive assignment, and otherwise to have one crossing if $y_i$ has a negative assignment. On all other gadgets, complete $M_{in}$ arbitrarily to an alternating cycle. Let $M_{tar}$ have a perfect matching on $Z$ such that $M_{tar}$ and $M_{in}$ differ on $Z$. Let $M_{tar}$ have its isolated vertex at the very end of $P$ and complete $M_{tar}$ with a perfect matching on the remainder of $P$ and $p$.
			
			Now take a flip sequence from $M_{in}$ to $M_{tar}$ of length at most $L+\ell+7(m_1+m_3)+8m_2+3K+2$. By our observations on lower bounds it follows that the flip sequence has to spend exactly seven flips on every $\exists$-gadget, eight flips on every $\forall$ gadget and three flips on every clause gadget. We reconstruct an assignment of $z_1,...,z_{m_3}$ from how their corresponding gadgets are switched. Since, every clause gadget has been flipped in the flip sequence, every clause contains at least one literal which has the right truth value assigned by the assignment from our constructions. Therefore, $\psi(x_1,...,x_{m_1},y_1,...,y_{m_2},z_1,...z_{m_3})=1$. Since we can do this for any assignments of $y_1,...,y_{m_2}$, this shows that $\psi$ is a \textsc{YES}-instance of \textsc{$\exists\forall\exists$-SAT}.
		\end{claimproof}
		
		The theorem then follows from the combination of the two claims.
	\end{proof}
	
	\section{Computing the Flip Distance is $\log$-APX-hard}\label{sec:apx}
	
	\begin{figure}[ht]
		\centering
		\includegraphics[scale=0.6]{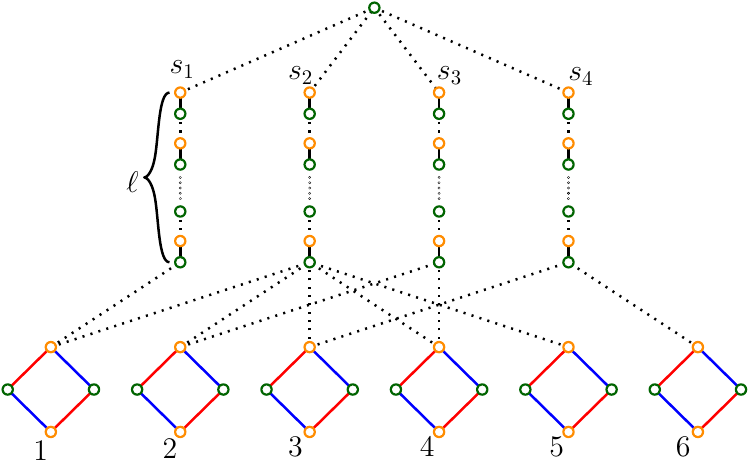}
		\caption{Instance for the hardness-proof. A \textsc{Set Cover} instance with $n=6$, $s_1= \{1\}$, $s_2=\{1,2,3,4,5\}$, $s_3=\{2,4\}$, and $s_4=\{3,6\}$ is reduced. Edges in $M_{in}\setminus M_{tar}$ are colored red, edges from $M_{tar}\setminus M_{in}$ blue, and edges in $M_{in}\cap M_{tar}$ black.}
		\label{fig:apx}
	\end{figure}
	
	The reduction from \textsc{Set Cover} to finding minimum flip sequences is illustrated in Figure~\ref{fig:apx}. For every set in $S$ from the \textsc{Set Cover} instance we add a path with~$\ell=6n+1$ vertices to the graph $G'''$. We construct $M_{in}$ and~$M_{tar}$ such that they form a perfect matching of the paths. For every integer from $1$ to $n$ (that denote the elements within the sets), we add a $4$-cycle to~$G'''$ such that $M_{in}$ and~$M_{tar}$ cover the cycle alternatingly. At last, we add one single vertex to~$G'''$ such that this vertex is the isolated vertex in both $M_{in}$ and $M_{tar}$. For each path corresponding to a set $s_i$ in $S$, we add an edge from one end of the path to the isolated vertex and edges from the other end to all gadgets that correspond to integers that are contained in $s_i$. The strategy for a flip sequence is to traverse a set gadget in order to be able to flip the integer gadgets for integers that are contained in the set. After that the flip sequence has to traverse the set gadget back to the initial isolated vertex. Since $G'''$ is bipartite (see green/orange partition in Figure~\ref{fig:apx}), there are no shortcuts that skip the traversal of set gadgets.
	
	\begin{restatable}{proposition}{NPH}
		\label{prop:np-hard}
		There exists a set cover of size $c$ if and only if $d(M_{in},M_{tar}) \leq c\cdot\ell + 3\cdot n$.
	\end{restatable}
	
	\begin{proof}
		Assume there exists a set cover of size $c$. To transform $M_{in}$ to $M_{tar}$ all cycles that correspond to an integer from $1$ to $n$ need to be switched. The only way to switch such a cycle is to flip along a path that corresponds to a set $s_i$ that is incident to $e$ to place the isolated vertex next to the integer gadget and then switch the cycle. At last, add all edges from the path that corresponds to $s_i$ back to the isolated vertex. If we repeat this for all vertices $s_i$ in a set cover, all cycles in integer gadgets will be switched in the end and we reached $M_{tar}$. Adding and removing the happy edges along a set gadget takes $\ell$ flips. Switching an integer gadget takes three flips. So the total length of the flip sequence is $c\cdot \ell + 3\cdot n$.
		
		For the opposite direction, we first need to see that we indeed have to take the detour back to the isolated vertex. Observe that~$G'''$ is a bipartite graph. In Figure \ref{fig:apx} we give a~$2$-coloring of $G'''$ with two colours, orange and green. The isolated vertex can then only be placed on green vertices. This prevents flip sequences from taking any shortcuts between gadgets, by placing the isolated vertex directly on a different path when leaving an integer gadget.
		
		Now assume that we have a flip sequence from $M_{in}$ to $M_{tar}$ of length at most $c\cdot\ell + 3\cdot n$. All cycles in integer gadgets have been switched. In particular, every cycle had the isolated vertex placed next to it at some point of the flip sequence. We set $S^\ast$ to be the set of all sets whose gadget contained the isolated vertex next to an integer gadget at some point. $S^\ast$ is clearly a set cover. By our initial observation, in order to place an isolated vertex in a set gadget next to an integer gadget, we need to charge at least $\ell$ flips towards the set gadget. Therefore, we get that $\lvert S^\ast\rvert \leq \frac{c\cdot\ell + 3\cdot n - 3\cdot n}{\ell} = c$  
	\end{proof}
	
	\APXH*
	
	\begin{proof}
		Assume we could approximate the length of a shortest flip sequence up to a factor of~$r$. Then, we can calculate a flip sequence of length $d\leq rd^\ast$ in polynomial time, where $d^\ast$ denotes the length of the optimal flip sequence. Further, let $c^\ast$ denote the size of an optimal set cover By Proposition~\ref{prop:np-hard} we can obtain a set cover of size
		
		\begin{equation*}
			c = \bigg\lfloor \frac{d- 3n}{\ell}\bigg\rfloor \leq \frac{rd^\ast -3n}{\ell} = \frac{rd^\ast -3r\cdot n + 3 r \cdot n - 3n}{\ell} = rc^\ast + \frac{r}{2}
		\end{equation*}
		
		Next, we need to rewrite the above expression to fit into the definition of an AP-reduction, that is, $c\leq c^\ast (1+\beta(r-1))$.
		
		First assume $(r-1)\geq \frac{1}{2c^\ast+1}$. We estimate (with $\beta = 4$):
		
		\begin{align*}
			c^\ast + \frac{r}{2} &= c^\ast + \frac{1}{2} + (r-1)c^\ast + \frac{r-1}{2} + (r-1)(3c^\ast-3c^\ast)\\
			&= c^\ast + 4(r-1)c^\ast + \frac{1}{2} - (r-1)(3c^\ast - \frac{1}{2}) \leq c^\ast(1+\beta(r-1))		
		\end{align*}
		
		We continue with the second case that $(r-1)<\frac{1}{2c^\ast+1}$.
		
		\begin{equation*}
			rc^\ast + \frac{r}{2} = c^\ast + \frac{1}{2} + (r-1)c^\ast + \frac{r-1}{2} < c^\ast + (r-1)(c\ast+\frac{1}{2}) + \frac{1}{2} < c^\ast + 1
		\end{equation*}
		In the second case, we recover the optimal solution. In the first case, the estimate fulfills the requirement for the AP-reduction.
	\end{proof}
	
	\section{Open Questions}
	We have shown that computing the diameter of the flip graph of combinatorial odd matchings is $\Pi_2^p$-complete and computing the radius of a the flip graph is $\Sigma_3^p$-complete. By doing so, we provide two naturally occuring problems that fall into these complexity classes. Further, by reducing \textsc{Set Cover} to the problem of finding shortest flip sequences between odd matchings we show that the problem is $\log$-\APX-hard and does not admit any global constant factor approximation.
	
	The following open questions arise from our results:
	
	\begin{enumerate}
		\item Since the reconfiguration of matchings is closely tied to complexity results in polytopes we motivate the question whether it is $\Sigma_3^p$-hard to compute the combinatorial radius/center of a polytope. We remark that the result for the diameter of the flip graph in~\cite{wulf2025computingpolytopediameterharder} uses a so called \emph{canonical structure}, that is, a structure that can be reached from any other structure in a reasonable number of flips. An upper bound on the diameter then follows from flipping from any initial structure to the canonical structure and then to the target structure. Therefore, by computing the radius, a center of the flip graph is implicitly also computed. This means that so far computing the radius is only shown to be $\Pi_2^p$-hard.
		\item Can similar results on the complexity of computing the diameter and radius also be shown in a geometric setting? While in the combinatorial setting, we can control, which edges are part of the input graph $G$ and which are not, in the geometric setting all edges between any two points in the plane can theoretically occur in $M_{in}$, $M_{tar}$ or any intermediate matching, the only degree of freedom for constructions is the placement of the points.
	\end{enumerate}
	
	\newpage
	
	\bibliographystyle{plainurl}
	\bibliography{citation}
	
	\newpage
	
		\section{Omitted Details in Section \ref{sec:diameter}}
	
	\connected*
	
	\begin{proof}
		By~\cite[Theorem~3]{aichholzer2025flippingoddmatchingsgeometric} the flip graph of odd matchings of $G'$ is connected if and only if for every edge $e \in E(G')$ one of the following holds:
		\begin{itemize}
			\item Every odd matching of $G'$ contains $e$.
			\item No odd matching of $G'$ contains $e$.
			\item Let $u_1$ and $u_2$ be the vertices incident to $e$. Either $G'-u_1$ or $G'-u_2$ contains a perfect matching.
		\end{itemize}
		We will show that the third condition holds for every edge based on a case distinction.
		
		\begin{description}\item[Case 1: $e$ contains $v$:] $G'-v$ can be partitioned into clause gadgets, $\forall$ gadgets, $\exists$ gadgets and $P$. All of them have a perfect matching. Therefore, $G'-v$ has a perfect matching.
			
			\item[Case 2: $e$ lies entirely on $P$:] Take an odd matching $M$ that has $v$ as its isolated vertex and take a path $\bar{P}$ in $G'$ that alternates between edges in $M$ and edges not in $M$ from $v$ to $u_1$. Without loss of generality, $\bar{P}$ is of even length. Otherwise, make $\bar{P}$ one edge longer or shorter to obtain a path of even length from $v$ to $u_2$. Then taking the symmetric difference $M\Delta E(\bar{P})$ is an odd matching with $u_1$ as isolated vertex.
			
			\item[Case 3: $e$ lies entirely on a variable gadget:] Take an odd matching $M$ that has $v$ as its isolated vertex. The variable gadget that contains $e$ is an even cycle with a perfect matching on it. We perform a flip sequence that switches the matching of the variable gadget. Along this flip sequence, every second vertex of the gadget in order of the cycle will be isolated. Therefore, there exists an intermediate matching of the flip sequence for which a vertex of $e$ is isolated.
			
			\begin{figure}[ht]
				\centering
				\includegraphics[scale=0.6]{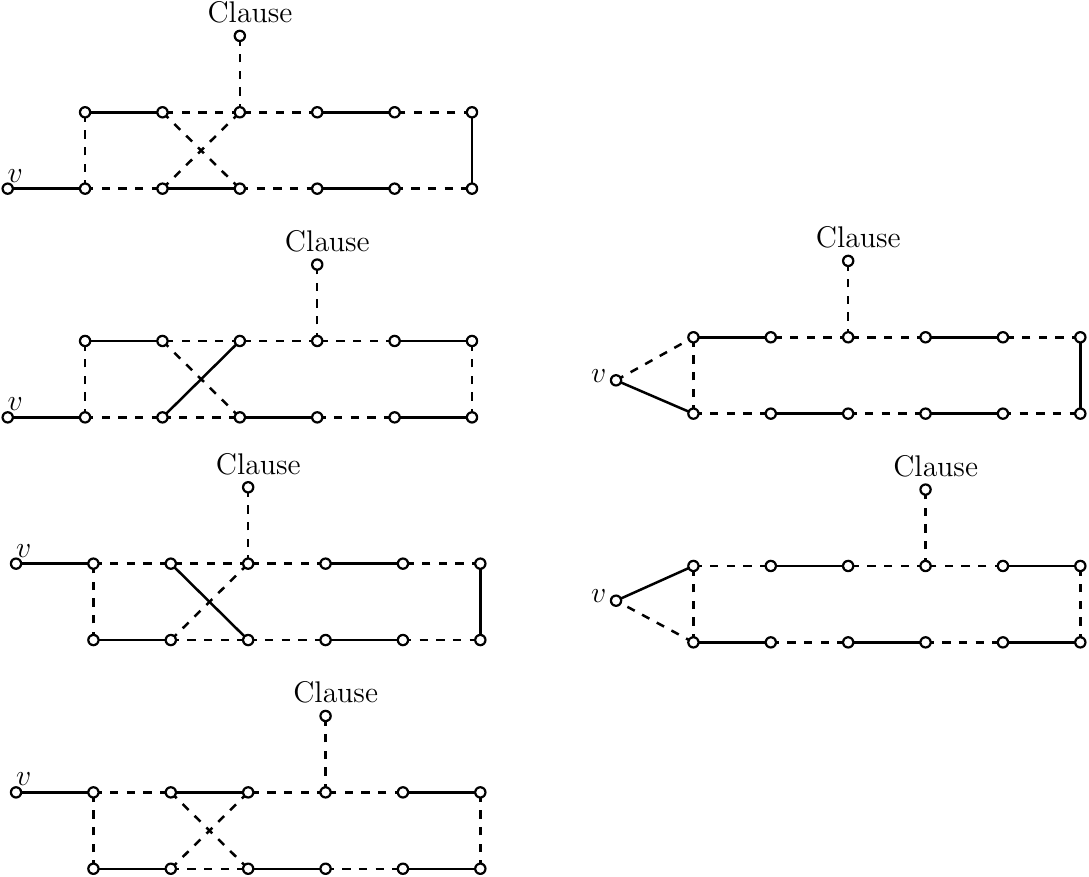}
				\caption{One possible way to complete a matching on a variable variable gadget with any vertex next to a clause gadget removed. Left: $\forall$-gadgets on one side, Right: $\exists$-gadgets.}
				\label{fig:matchings}
			\end{figure} 
			
			\item[Case 4: $e$ connects a vertex of a variable gadget and a vertex of a clause gadget:] Let $u_1$ be the vertex of $e$ that belongs to a variable gadget. $G'-u_1$ contains a perfect matching that has an induced perfect matching on all clause gadgets, on all variable gadgets except the one containing $u_1$ and a perfect matching on $P$. All variants for the variable gadget that contains $u_1$ are covered in Figure \ref{fig:matchings}.
			
			\item[Case 5: $e$ lies entirely on a clause gadget:] Take an odd matching $M$ that has an isolated vertex placed on a vertex of a variable gadget next to the clause gadget as constructed in Case 4. The clause gadget that contains $e$ is an even cycle with a perfect matching on it. We perform a flip sequence that switches the matching of the variable gadget. Along this flip sequence, every second vertex of the gadget in order along the cycle will be isolated. Therefore, there exists an intermediate matching of the flip sequence for which a vertex of $e$ is isolated.
		\end{description}  
	\end{proof}
	
	\structure*
	
	\begin{proof}
		{\textbf{\sffamily{(1)}}} Assume neither $v_{M_{in}}$ nor $v_{M_{tar}}$ lies on $P$, in particular the matchings coincide on $P$. We consider the flip graph of odd matchings on $G'-P$. Any flip sequence from $M_{in}\setminus E(P)$ to $M_{tar}\setminus E(P)$ on $G'-P$ can be embedded to be a flip sequence from $M_{in}$ to $M_{tar}$ in $G'$. The diameter of the flip graph of odd matchings of $G'-P$ is bounded by $c\cdot\lvert V(G'-P)\rvert = c\cdot (s+1)$ and, therefore, $d(M_{in},M_{tar})\leq c\cdot(s+1)$. By the choice of the length $\ell$ simply moving the isolated vertex from $v$ to $w$ takes more flips than that. Therefore, $M_{in}$ and $M_{tar}$ cannot form a maximal pair. If only one of the two matchings has the isolated vertex on $P$ it takes at most $c\cdot(s+1)$ flips to move the isolated vertex away from $P$. Whereas it takes at least $2c\cdot(s+1)$ flips to move the isolated vertex from $v$ to $w$ and back. We conclude that both matchings have the isolated vertex on $P$. Additionally, if for example $M_{in}$ has the isolated vertex on $P$, but not at~$w$, a matching that coincides with $M_{in}$ on all the variable and clause gadgets, but has the isolated vertex closer to $w$ requires strictly more flips to move the isolated vertex to $v$.
		
		{\textbf{\sffamily{(2)}}} Since by (1) $M_{in}$ and $M_{tar}$ have their isolated vertex placed at $w$, the two matchings contain a perfect matching of the vertices of every clause gadget. Note, that if we place the isolated vertex on the clause gadget at some point by adding an edge between a variable gadget and the clause gadget, after removing said edge, the isolated vertex will again be in the same variable gadget. If we remove all flips from a flip sequence that modify clause gadgets which contain happy edges, we obtain a shorter flip sequence with the same initial and target matchings. Taking the same matchings $M_{in}$ and $M_{tar}$ and only changing one of the two matchings on the clause gadget to get an alternating cylce will lead to a larger flip distance. Any shortest flip sequence in the altered matching will take three additional flips on the clause gadget. A shorter flip sequence between the altered matchings would have to locally do better on some other gadget and, therefore in the same gadget with the same set of edges in $M_{in}$ and $M_{tar}$.
		
		{\textbf{\sffamily{(3)}}} In Figure \ref{fig:estates} we portray all possible ways that $M_{in} \cup M_{tar}$ can look like on an $\exists$-gadget. Alternating cycles as seen on the top left and the bottom right can be switched in seven flips. Depending on which edge from $v$ to the gadget is added, the isolated vertices will be placed, such that they are either neighbors to clause gadgets that contain the positive form of $y_i$ or such that they are neighbors to the clause gadgets containing the negative form of $y_i$.
		
		If, however, the variable gadget contains six happy edges, we can place isolated vertices next to clauses containing $y_i$ or $\bar{y_i}$ by removing two happy edges. Which of the two values occurs, again, depends on which edge from $v$ is added to enter the gadget. In this case, we only have to charge a total of four flips towards the gadget to reach the same clause gadgets for which we previously used seven flips.
		
		We conclude that for every pair of matchings that has an alternating cycle in an $\exists$-gadget, a pair of matchings that has happy edges in the same variable gadget can be flipped into one another with three flips less.
		
		\begin{figure}[ht]
			\centering
			\includegraphics[scale=0.6]{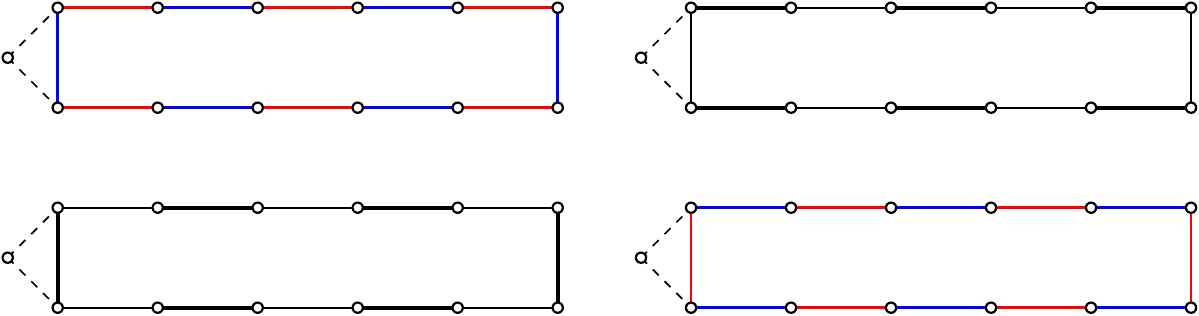}
			\caption{All possible states of the union of $M_{in} \cup M_{tar}$ in an $\exists$-gadget. Edges from $M_{in}\setminus M_{tar}$ are blue, from $M_{tar}\setminus M_{in}$ red and from $M_{in}\cap M_{tar}$ black.}
			\label{fig:estates}
		\end{figure}
		
		{\textbf{\sffamily{(4)}}} In Figure \ref{fig:foralls}, we show all possible states of $M_{in}\cup M_{tar}$ on a $\forall$-gadget corresponding to $x_i$. If $M_{in}\cup M_{tar}$ forms one big alternating cycle with twelve edges, it takes seven flips to switch the gadget. If the cycle is not crossing, the isolated vertex will be placed next to clause gadgets that contain the positive literal $x_i$. If the cycle is crossing, the isolated vertex will be placed next to clause gadgets that contain the negative literal~$\bar{x}_i$. We conclude that we can reach clauses with one of the two truth assignments in seven flips. Placing the isolated vertex next to a clause gadget of the other truth value takes at two flips more, that is, nine flips.
		
		For every case where $M_{in}\cup M_{tar}$ forms a matching with twelve happy edges, there exists a way to remove two happy edges to place the isolated vertex next to clause gadgets of one truth value. Further, we can remove one additional happy edge to place the isolated vertex next to the clause gadget of the other truth value. We conclude that we can reach clauses of one truth value in four flips and clauses with both truth values in $6$ flips.
		
		In the last cases, where $M_{in}\cup M_{tar}$ consists of four happy edges and a small alternating cycle, we need to remove one happy edge in order to then switch the cycle of length four. In the process, the isolated vertex will be placed next to all clause gadgets that contain any literal $x_i$ or $\bar{x}_i$. In these cases, we charge five flips towards the variable gadget.
		
		We conclude, that the extremal cases are the ones, where $M_{in}\cup M_{tar}$ form one big alternating cycle. Whether we want to reach the clause gadgets that we can reach with fewer flips, or the clause gadgets that we can reach with more flips, or clause gadgets with both truth values, in all three cases the big alternating cycle takes the requires the most flips.
		
		\begin{figure}[ht]
			\centering
			\includegraphics[scale=0.6]{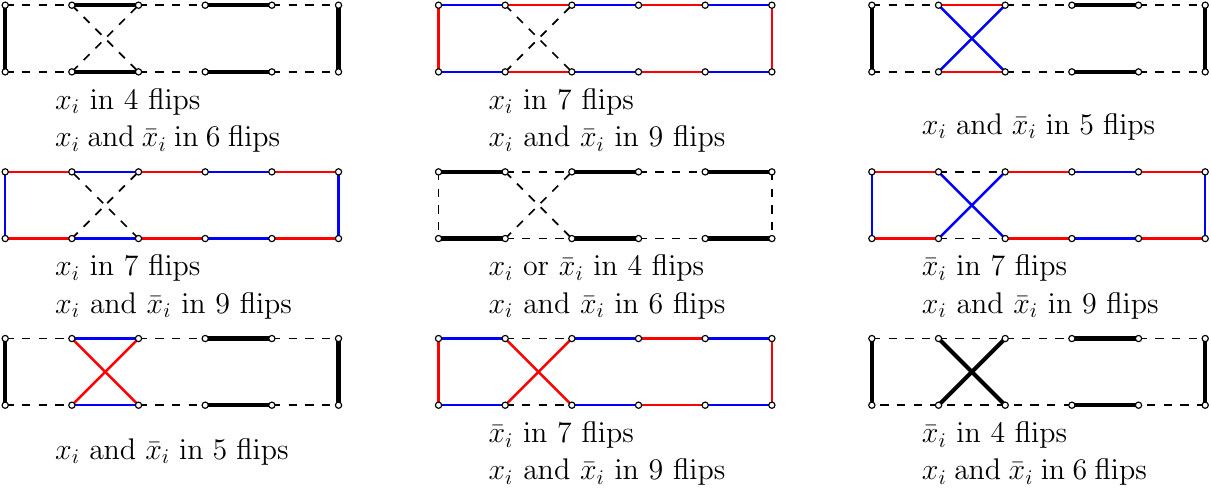}
			\caption{All possible states of the union of $M_{in}$ (blue) and $M_{tar}$ (red) in a $\forall$-gadget.}
			\label{fig:foralls}
		\end{figure}
		
	\end{proof}
	
	\newpage
	
	\section{Omitted Details in Section \ref{sec:radius}}
	
	\connecteds*
	
	\begin{proof}
		The proof is mostly the same as the one of Proposition \ref{prop:connected}. The only thing that changes, is that in Case 4, the gadget can also be one of the new $\forall$-gadgets. In Figure \ref{fig:mat2}, for every removed vertex next to a clause gadget, we can obtain a perfect matching. The connectivity of the flip graph again follows from \cite[Theorem~3]{aichholzer2025flippingoddmatchingsgeometric}.  
		
		\begin{figure}[ht]
			\centering
			\includegraphics[scale=0.6]{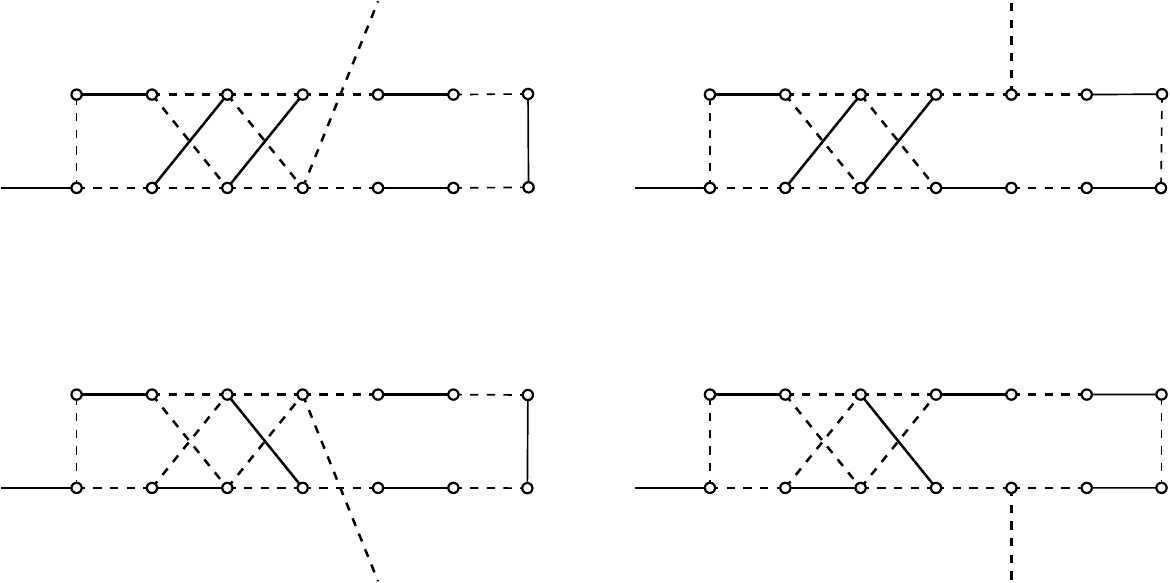}
			\caption{How to complete $\forall$-gadgets to perfect matchings, when one vertex is isolated.}
			\label{fig:mat2}
		\end{figure}
	\end{proof}
	
	\structural*
	
	\begin{proof}
		{\textbf{\sffamily{(1)}}} If $v_{M_{tar}}$ lies on $P$, then consider the flip graph of odd matchings of $G''-P$. Any flip sequence from $M_{in}\setminus E(P)$ to $M_{tar}\setminus E(P)$ on $G''-P$ can be embedded into a flip sequence from $M_{in}$ to $M_{tar}$ in $G''$. The diameter of the flip graph of odd matchings is bounded by $c\cdot\lvert V(G''-P)\rvert = c\cdot (s+\ell+1)$ and therefore $d(M_{in},M_{tar})\leq c\cdot(s+\ell+1)$. By the choice of the length $L$ simply moving the isolated vertex from one end of $P$ to the other takes more flips than that. Therefore, $M_{tar}$ will increase the distance from $M_{in}$ by simply moving its isolated vertex there. Therefore,~$M_{tar}$ has its isolated vertex on $P$. Additionally, if~$M_{tar}$ has the isolated vertex on $P$, but not the  farthest away from $v$, a matching that coincides with $M_{tar}$ on all the variable and clause gadgets, but has the isolated vertex further away from $v$ has larger distance to $M_{in}$.
		
		{\textbf{\sffamily{(2)}}} $v_{M_{in}}$ will not be placed on $P$, because if $M_{tar}$ differs from~$M_{in}$ on any clause or variable gadget, a flip sequence would have to start by putting the isolated vertex to $v$, which can be avoided by placing the isolated vertex not on~$P$ in the first place. Secondly, if $v_{M_{in}}$ is not on $p$ and~$M_{tar}$ differs from $M_{in}$ on $Z$ a flip sequence would have to move the isolated vertex along $p$ to switch $Z$ and back, which requires a total of $\frac{\ell}{2}$ flips, in each direction. By the choice of $\frac{\ell}{2}$ this accounts for more flips than flipping all of $G''-P-p$ together, in each direction. Therefore, to safe $\ell$ flips, $v_{M_{in}}$ will not be placed outside of $p$. Additionally, if~$M_{in}$ has the isolated vertices on $p$, but not on $Z$, a matching that coincides with~$M_{in}$ on all the variable and clause gadgets, but has the isolated vertex closer to~$Z$ has less distance to any~$M_{tar}$ of maximal distance to $M_{in}$.
		
		{\textbf{\sffamily{(3)}}} Follows analogously to Case (2) in Lemma \ref{lem:structure}. To see this, note that first~$M_{in}$ is fixed and only then $M_{tar}$ gets picked to maximize the distance from $M_{in}$.
		
		{\textbf{\sffamily{(4)}}} For the Second $\exists$-gadget, we refer to Figure \ref{fig:estates} and the analysis thereof in Case (3) of Lemma \ref{lem:structure}. $M_{in}$ will choose one perfect matching of the gadget and~$M_{tar}$ will choose the other perfect matching in order to make any shortest flip sequence charge seven flips towards the gadget instead of just four flips. The same holds for the First $\exists$-gadget, which has been analyzed in Case (4) of Lemma~\ref{lem:structure} with the help of Figure \ref{fig:foralls}.
		
		It remains to analyze the $\forall$-gadget. All possible unions of $M_{in}$ and $M_{tar}$ are depicted in Figure \ref{fig:sforall2}. The top row shows all alternating cycles with exactly one crossing. Switching those needs eight flips and in the process the isolated vertex will be placed next to clauses that contain the negated literal $\bar{y}$. Further, we can reach clauses that contain $y$ be spending an addtional two flips. The second row shows all alternating cycles with zero or two crossings. Switching those cycles requires eight flips and in the process the isolated vertex will be placed next to clauses that contain the positive literal $y$. Further, $\bar{y}$ can be reached by spending an additional two flips.
		
		In the bottom row all cases where $M_{in} \cup M_{tar}$ consists of only happy edges are depicted. In all of the cases, the isolated vertex can either be placed next to clauses containing the literal $y$ or $\bar{y}$ and this happens in four to six flips. Again, with additional two flips, clauses with variables of the other truth value can be reached. Further observe, that the choice of $M_{tar}$ does not change which clauses can be switched, it only changes the amount of flips needed to do so. Therefore, states from the bottom row will not occur in a pair of matchings realizes the radius, since they require less flips than the states in the first two rows while simultaneously granting more choice with regard to the assignment of the variable.
		
		The third row contains all cases where $M_{in} \cup M_{tar}$ consists of four happy edges and an alternating four-cycle. In all of these cases, the isolated vertex can either be placed next to clauses containing the literal $y$ or $\bar{y}$ and this happens in five to seven flips, two for removing and adding a happy edge, three for switching the cycle and possibly another two for removing and adding another happy edge. By spending two additional flips, or non at all, we can place the isolated vertex next to clauses of that contain literals of the other truth value as well. Further observe that the choice of $M_{tar}$ does not change which clauses can be switched. It only changes the amount of flips needed to do so. Therefore, states from the third row will also not occur in a pair of matchings that realizes the radius, since they require less flips than the states in the first two while simultaneously granting more choice with regard to the assignment of the variable.  
		
		\begin{figure}[ht]
			\centering
			\includegraphics[scale=0.6]{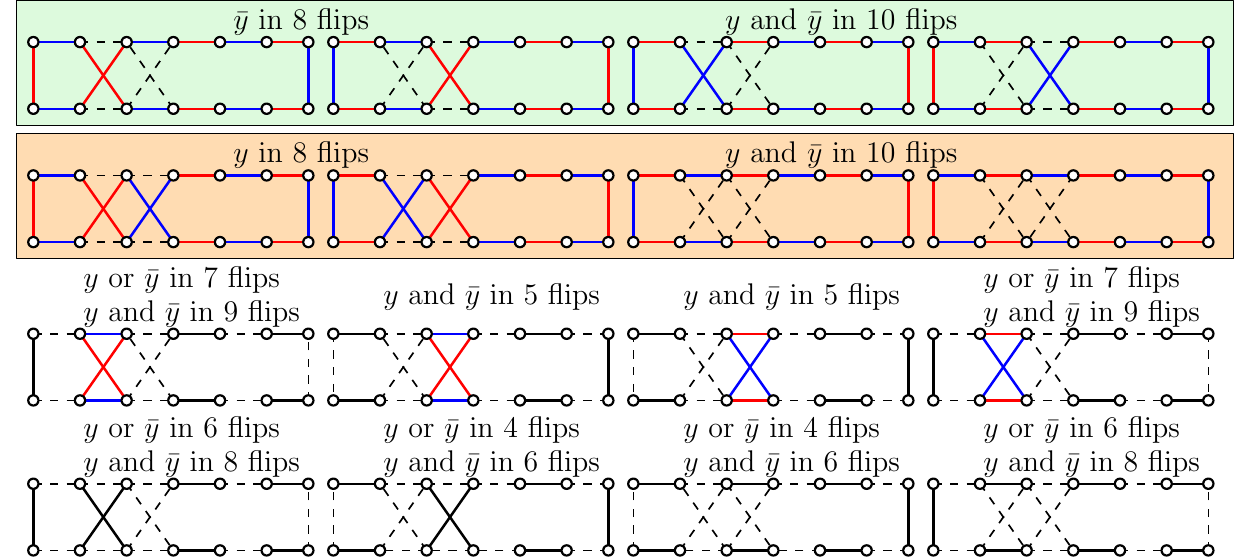}
			\caption{All states of the $\forall$-gadget.}
			\label{fig:sforall2}
		\end{figure}
	\end{proof}
	
\end{document}